\documentclass[letter paper, 10 pt, conference]{ieeeconf}

\IEEEoverridecommandlockouts                              % This command is only needed if 
\usepackage{textcomp}

\usepackage{epsfig} % for postscript graphics files
\usepackage{psfrag}
\usepackage{cite}
\usepackage{amsmath} % assumes amsmath package installed
\usepackage{amssymb}  % assumes amsmath package installed
\usepackage{dsfont}  % assumes amsmath package installed
\usepackage{todonotes}
\usepackage{graphicx}
\usepackage{subcaption}
\usepackage[noend]{algorithmic}
\usepackage{algorithm, ulem}
\usepackage{hyperref}
\usepackage{cleveref}
\usepackage{enumitem}

\newcommand{\tr}{\mathrm{tr}}
\newcommand{\bbP}{\mathbb{P}}
\newcommand{\E}{\mathbb{E}}
\newcommand{\xc}{x^{\rm c}}
\newcommand{\xs}{x^{\rm s}}
\newcommand{\ec}{e^{\rm c}}

\newcommand{\xhat}{{x}^{\rm a}}
\newcommand{\xihat}{{\xi}^{\rm a}}
\newcommand{\zetahat}{{\zeta}^{\rm a}}
\newcommand{\za}{z^{\rm a}}
\newcommand{\qq}{q}
\DeclareRobustCommand{\bigO}{%
  \text{\usefont{OMS}{cmsy}{m}{n}O}%
}

\usepackage{color}
\definecolor{red}{rgb}{1,0.2,0.2}
\definecolor{green}{rgb}{0.2,1,0.5}
\definecolor{blue}{rgb}{0,0,1}
\definecolor{lightblue}{rgb}{0.3,0.5,1}

\usepackage{amsthm}
\newtheorem{proposition}{Proposition}
\newtheorem{lemma}{Lemma}
\newtheorem{theorem}{Theorem}

\theoremstyle{remark}
\newtheorem{remark}{Remark}

\title{On Model Protection in Federated Learning against Eavesdropping Attacks} 

\author{ Dipankar Maity and Kushal Chakrabarti
\thanks{D. Maity is with the Department of Electrical and Computer Engineering and an affiliated faculty of the North Carolina Battery Complexity, Autonomous Vehicle, and Electrification Research Center (BATT CAVE), University of North Carolina at Charlotte,  NC, 28223, USA (e-mail: {\tt {dmaity@charlotte.edu}}).}
\thanks{K. Chakrabarti is with the Data and Decision Sciences division of Tata Consultancy Services Research, Mumbai, India (e-mail: {\tt {chakrabarti.k1@tcs.com}}).}
}

\begin{document}

\maketitle

\begin{abstract}
In this study, we investigate the protection offered by Federated Learning algorithms against eavesdropping adversaries. In our model, the adversary is capable of intercepting model updates transmitted from clients to the server, enabling
it to create its own estimate of the model. Unlike previous research, which predominantly focuses on safeguarding client data, our work shifts attention to protecting the client model itself. Through a theoretical analysis, we examine how various factors—such as the probability of client selection, the structure
of local objective functions, global aggregation at the server, and the eavesdropper’s capabilities—impact the overall level of protection. We further validate our findings through numerical experiments, assessing the protection by evaluating the model accuracy achieved by the adversary. Finally, we compare our
results with methods based on differential privacy, underscoring their limitations in this specific context.
\end{abstract} 

\section{Introduction} 
\label{sec:intro}

Traditionally, deep learning techniques require centralized data collection and processing that may be infeasible in collaborative scenarios, such as healthcare, credit scoring, vehicle fleet learning, internet-of-things, e-commerce, and natural language processing, due to the high scalability of modern networks, growing sensitive data privacy concerns, and legal regulations such as GDPR~\cite{fu2024differentially, cummings2018role, nguyen2022federated}. In these domains, data is often distributed among multiple parties of interest, with no single trusted authority. Federated Learning (FL) has emerged as a distributed collaborative learning paradigm, which allows coordination among multiple clients to perform training without sharing raw data. Instead, they participate in the learning process by training models locally and sharing only the model parameters with a central server. This server aggregates the updates and redistributes the improved model to all participants~\cite{mcmahan2017communication, bonawitz2019towards}. Based on the distribution/partition of data among the clients, FL can be classified into horizontal (HFL), vertical (VFL), and transfer (TFL) federated learning~\cite{zhang2021survey, fu2024differentially}. In HFL, the clients have the same feature space but different sample spaces. So, the HFL clients can adopt the same global model for training their datasets. In this paper, we focus on HFL.

Despite the intent to protect sensitive information about the learning process through FL, the clients' raw data and the model parameters are not completely secure during the communications rounds in the training phase. 
Adversaries may still infer sensitive data from the shared model parameters through model inversion attacks~\cite{zhu2019deep, zhao2020idlg, geiping2020inverting, wilson2024federated, lyu2022privacy} or inference attacks~\cite{nasr2019comprehensive, mironov2017renyi}. Cryptographic methods such as homomorphic encryption~\cite{xie2024efficiency} or secure multiparty computation~\cite{li2020privacy} significantly enhance privacy in FL. However, these techniques require significant computational and communication overheads, rendering them unsuitable for client devices with limited power (e.g., mobile phones) or limited bandwidth. 
Differential privacy is a state-of-the-art framework to reduce privacy risks in FL. Gradient clipping and noising the local model parameters before sharing them with the server can reduce the efficiency of the model inversion attacks~\cite{zhu2019deep, wei2021gradient}. While differential privacy offers provable guarantees, it also reduces the accuracy of the federated learning models~\cite{wei2020federated, zhu2019deep}. A hardware-based defense framework is Trusted Execution Environments (TEEs), e.g., Intel SGX, which are secure areas within a device that ensure computations are isolated from other processes and protected from unauthorized access, providing a safe environment for privacy-sensitive operations~\cite{zhang2020enabling}. However, TEEs may have limited computational resources and can be vulnerable to certain side-channel attacks, reducing their overall security and scalability~\cite{zhang2022security, nilsson2020survey}.

While there is a plethora of works on protecting clients' raw data in FL~\cite{li2021survey, bhowmick2018protection, wen2023survey, xu2019hybridalpha, zhang2022security}, the open literature has given relatively less emphasis on protecting the clients' models or the global model during the training process. Protecting the model parameters is also important because it represents the solution
to a certain decision-making problem. An adversary having access to the model parameters or its proxy may cause potential privacy threats. The adversary could use it for further model poisoning attacks to mislead the global model~\cite{tolpegin2020data}, for competitive purposes~\cite{tramer2016stealing}, or they may be able to find ways to reverse-engineer privacy-preserving techniques such as differential privacy~\cite{shokri2017membership}. In this paper we investigate this unexplored area in FL---protection of the clients' model during training. 

While the majority of the existing aforementioned works on client's data privacy consider either the server or a client to be honest but curious, i.e., interested in reconstructing the peer clients' training data, less focus has been on a third-party adversary, who simply listens to the   exchanged messages (i.e., an eavesdropping adversary~\cite{wang2019privacy}) between the client and the server.
In this initial work, we consider a setting where the adversary eavesdrops on the uplink messages of a specific client with the goal of learning that client’s local model. 
% Moreover, in contrast to an actively participating adversary (e.g., honest but curious client), to protect information from eavesdropping adversaries, one needs to deploy a mechanism where the eavesdroppers cannot decode the information from the
% intercepted communications perfectly, and thus, it becomes more challenging to protect against these adversaries. Thus, in this work, we consider such an honest but curious adversary who is eavesdropping the communication channel from a client to the server. \TODO{Mention only uplink monitoring and its justification.} \dm{try it}

The existing HFL frameworks have two main types of client-to-server communication models: (i) clients sending their local models to the server, as used in FedAvg~\cite{mcmahan2017communication}, FedProx~\cite{li2020federated}, FedDyn~\cite{acarfederated}, FedDANE~\cite{li2019feddane}, FedDC~\cite{gao2022feddc} algorithms, and (ii) clients sending the differences between their updated local model and the one received from the server, as in MOCHA~\cite{smith2017federated}, FL vis STC~\cite{sattler2019robust}, SCAFFOLD~\cite{karimireddy2020scaffold}, FSVRG~\cite{konevcny2016federated} algorithms.
The primary objective of this work is to analyze and quantify the level of protection offered in both cases. We emphasize that our goal is \textit{not} to develop new “protected” algorithms but rather to characterize the amount of protection existing algorithms provide.
We view the model's \textit{protection} as a system-level property that each dynamical system possesses \cite{maity2024ensuring}, and in this work, we analyze and compute this protection.

While \textit{accuracy} and \textit{convergence rate} have traditionally been the dominant metrics for selecting among FL algorithms, we argue that model's \textit{protection} is another critical factor that warrants consideration, as we will show that some very powerful algorithms can be \textit{0-protected}. The main contribution of our work is developing a framework, offering a principled method to analyze and compute this model protection, thereby enabling a principled protection comparison across different FL algorithms. Our key contributions are summarized as follows.
\begin{itemize}[leftmargin=*]
    \item We show that the class of FL algorithms, notably FedAvg, FedProx, FedDyn, FedDANE, and FedDC, where the clients upload their local models to the server, are \textit{0-protected}. 
    \item We analytically compute the protection offered by a class of FL algorithms in which clients share model increments with the server.
    Our analysis highlights how client selection probability, adversary capability, and algorithm design influence the level of protection.
    Notably, the analysis reveals that these algorithms retain a non-zero level of protection even when the adversary can eavesdrop with probability 1.
    \item We numerically demonstrate that uploading local model updates, rather than the entire local model, provides significantly stronger protection against eavesdropping adversaries in FL. Our experimental results, conducted on the CIFAR-10 benchmark dataset, validate the theoretical claim that this approach enhances the privacy of client model parameters. To contextualize our framework within the existing literature on client data protection, we compare our method with the FedAvg algorithm, which involves clients uploading their full model to the server, supplemented with a differential privacy mechanism (FedAvg+DP). The adversary's (estimated) model evaluation on reconstructed client data, obtained by the deep leakage attack~\cite{zhu2019deep}, shows that uploading model updates offers superior protection compared to FedAvg+DP, strengthening the FL process against privacy leaks.
\end{itemize}

% \dm{what is the appropriate place to put the figure below?}\kc{Maybe first page right side, or before Section II-A.}
% The interaction between the client and the server can be abstracted by the following model in \Cref{fig:abstracted}.
% \begin{figure}[h]
%     \centering
%     \includegraphics[width=\linewidth]{Figures/Setup.png}
%     \caption{\TODO{Figure to be updated. Powerpoint file is available in the folder.}}
%     \label{fig:abstracted}
% \end{figure}

{\bf Notation}: 
Throughout this paper, we use the superscripts $(\cdot)^a$, $(\cdot)^c$, and $(\cdot)^s$ to denote quantities associated to the adversary, client, and the server, respectively. We let $I$ denote the identity matrix of the appropriate dimension. For a vector or matrix, $(\cdot)^\intercal$ denotes its transpose. Let $N$ be any natural number.
We use $\nabla g(x)$ and $\nabla^2 g(x)$ to denote the gradient and Hessian of a function $g: \mathbb{R}^N \to \mathbb{R}$ evaluated at $x \in \mathbb{R}^N$. We let $\|\cdot\|$ denote the Euclidean norm of a vector and $|\cdot|$ denote the absolute value of a scalar. 
We let $\tr(\cdot)$ and $(\cdot)^{-1}$ denote the trace and inverse of a matrix, respectively. 
We let $\lambda_{\max}(\cdot)$ denote the maximum magnitude of eigenvalues of a matrix. For two symmetric matrices $A,B$ of same dimensions, we let $A \succeq B$ denote that $(A-B)$ is symmetric positive semi-definite. Finally, $\max\{a,b\}$ denote the maximum of two real-valued scalars $a$ and $b$.

% \red{Testing bibliography: \cite{maity2024ensuring}}

\section{Problem Formulation} 

We consider an FL problem in which the uplink from a client to the server is compromised.
An adversary eavesdrops on this communication and attempts to estimate the client's model parameters.  
Let $\xc_t$ and $\xhat_t$ denote the client's model and the adversary's estimate, respectively, at global time $t$. 
In this work, we adopt the squared error norm as a metric of protection  against eavesdropping adversaries; see \cite{agarwal2020distortion, arXivMaity2021} for details. 
Specifically, the \textit{protection} at time $t$ is quantified by $\E[\|\xc_t - \xhat_t\|^2]$ and the final (asymptotic) protection by $\lim_{t\to \infty} \E[\|\xc_t - \xhat_t\|^2]$. 
The expectation is taken with respect to the randomness in both the eavesdropping process and the client selection of the FL algorithm.

\subsection{Federated Learning Architecture}
\label{sub:fl_arch}

In a typical FL setup, the server (randomly) samples a subset of $n$ clients at each global round $t$ and broadcast the global model $\xs_t$ to them.
Each client independently performs stochastic gradient-descent (SGD) on the received server model and returns the updated local model $\xc_{t+1}$ (or the model difference $\xc_{t+1}-\xs_t$) to the server.
Let the random variable $\delta_t \in \{0,1\}$ denote whether the compromised client is sampled at round $t$ (i.e., $\delta_t =1$) or not (i.e., $\delta_t =0$). 
Assuming clients are sampled uniform randomly, we have 
\begin{align*}
    \bbP(\delta_t = 1) = \frac{n}{N} := p,
\end{align*}
where $N$ is the total number of clients. We assume that client sampling is independent at each round, i.e., $\delta_t$ and $\delta_{t'}$ are independent random variables for all $t\ne t'$.

Let us define 

\begin{align*}
    \xi_{t} = \xc_{t+1} - \xs_{t},
\end{align*}
to be the model update by the client resulting from its $K$ steps of gradient-descent (GD).\footnote{For the simplicity of the exposition, we are considering gradient-descent instead of stochastic gradient-descent. However, the analysis can be carried out in a very similar fashion for stochastic gradient descents.} 
That is,
\begin{align} \label{eq:xi_t}
    \xi_t &= - \eta\sum_{k=0}^K \nabla f^c (x_k),\\
    x_{k+1} &= x_k - \eta \nabla f^c (x_k), ~ \forall k \in [0,K-1], \quad x_0 = \xs_t. \nonumber
\end{align}
Here, $\eta$ and $f^c$ denote the client's learning rate and local objective function, respectively. In controls, signals akin to $\xi_t$ are sometimes referred to as the \textit{innovation signals} or \textit{new information}, especially when they have randomness in them and when they are used in estimating certain quantities.

The dynamics of the client's model can be written as 
\begin{align*}
    \xc_{t+1} = \begin{cases}
        \xs_t + \xi_t \qquad &\delta_t = 1\\
        \xc_t, &\delta_t = 0. 
    \end{cases}
\end{align*}
In a more compact form, we may write, 
\begin{align} 
    \xc_{t+1} & = \xc_t + \delta_t (\xi_t + \zeta_t),  \label{eq:client_dyn} \\
    \zeta_t & = \xs_t - \xc_t,  \label{eq:zeta_t}
\end{align}
where $\zeta_t$ is the model mismatch between the client and server at time $t$. 

For ease of reference, we refer to the above FL algorithm as Federated Learning with Innovation Process (FLIP), when the client returns $x^c_{t+1}-x^s_t$ to the server, and as Federated Learning withOut innovation Process (FLOP), when the client returns $x^c_{t+1}$. Since we aim to differentiate the \textit{protection} of FLIP from FLOP, we considered only GD updates for local optimization in~\eqref{eq:xi_t} for simplicity of presentation. As the protection analysis, presented later in Section~\ref{sec:prot_analysis}, is agnostic to the specific form of $\xi_t$ in~\eqref{eq:xi_t},
our framework can be extended to other FL algorithms mentioned earlier in Section~\ref{sec:intro}. 
% \kc{Any better name than IS-FL?} \dm{How about {\tt iFL} or Federated Learning with Innovation Process {\tt FLIP}?}

% \subsection{Adversary Eavesdropping Model}

\subsection{Adversary Eavesdropping Model \& Estimation Dynamics}
\label{sub:adv_model}

The eavesdropping mechanism is modeled by a Bernoulli process $\mu_t \sim \text{Bernoulli}(\gamma)$ \cite{tsiamis2019state}. 
A successful eavesdropping of the client-to-server communication at global round $t$ is denoted by the event $\mu_t = 1$.
Likewise, $\mu_t = 0$ denotes an unsuccessful eavesdropping outcome. 
We assume that the adversary knows whether the client is sampled at the current time or not. 
That is, it knows the outcome of $\delta_t$. We also assume that the eavesdropping process is temporally uncorrelated—i.e., $\mu_t$ and $\mu_{t'}$ are independent for all $t \ne t'$—and independent of the client selection process, i.e., $\mu_t$ and $\delta_{t'}$ are independent random variables for all $t$ and $t'$.

We let the adversary follow the estimation dynamics 
\begin{subequations} \label{eq:adv_dynamics}
\begin{align}
    \xhat_{t+1} & = \xhat_t + \delta_t \za_{t}, \label{eqn:xhat_def} \\
   \za_{t} & = (\mu_{t} \xi_{t} + (1- \mu_{t})  \xihat_t) + \zetahat_t, \label{eqn:xihat_def}
\end{align}
\end{subequations}
which mimics the client dynamics \eqref{eq:client_dyn} and computes an estimate of $\xi_t + \zeta_t$ via $z_t$. 
Notice that, when eavesdropping is successful, the adversary only intercepts the uplink message $\xi_t$ and not the model difference $\zeta_t$.
Hence, it computes an estimate for $\zeta_t$ as $\zetahat_t$ in \eqref{eqn:xihat_def}.
On the other hand, when eavesdropping fails, it estimates both $\xi_t$ and $\zeta_t$ in \eqref{eqn:xihat_def}.
Finally, $\delta_t = 0$ yields $\xhat_{t+1} = \xhat_t $, making it consistent with the client dynamics $\xc_{t+1} = \xc_t $ for $\delta_t=0$.

\begin{remark} \label{rem:zeta}
    Notice that, if $t'$ is the most recent round before $t$ when the client was selected, then $\zeta_t = \xs_t - \xc_t = \xs_t - \xs_{t'} - \xi_{t'}$.
   Therefore, computing $\zeta_t$ requires knowledge of the server state change, $\xs_t - \xs_{t'}$, over the interval $[t', t]$ even when the adversary is able to intercept $\xi_{t'}$. This quantity is a random variable influenced by several factors, including the client selection mechanism, the local objective functions, and the data distributions at other clients. These dependencies make accurately estimating $\zeta_t$ particularly challenging.
    In this work, we do not focus on any specific estimation technique for $\zetahat_t$, and leave it for a future work.
    The final expression of protection will show how this estimate, $\zetahat_t$, affects protection. \hfill $\blacksquare$
\end{remark}

% Before proceeding further, we provide a brief justification for this model. 
% In particular, we prove that both $\zeta_t$ and $\xi_t$ follow the dynamics $\zeta_{t} \approx G \zeta_{t'}$ and $\xi_t \approx H_1 \xi_{t'} + H_2 \zeta_{t'}$, where ${t'} < t$ is the latest round the client was sampled before being sampled again at round $t$. 
% To this end, let us first start with the proof of $\zeta_t \approx G \zeta_{t'}$. 

To develop an estimate for $\xi_t$, we proceed as follows. Recall that, $k$ denote the local steps index at the client for a given $t$.
Let us define $y_k := \nabla f^c (x_k)$ and obtain
\begin{align*}
    y_{k+1} & = \nabla f^c (x_{k+1}) = \nabla f^c (x_k - \eta y_k) \\
    & \approx \nabla f^c (x_k) - [\nabla^2 f^c (x_k)] \eta y_k \\
    & = \underbrace{\!\!(I - \eta \nabla^2 f^c (x_k))\!\!}_{:=F_k} ~~y_k \\
    & = \big( \prod\nolimits_{\ell = 0}^k F_\ell) y_0 =  \big( \prod\nolimits_{\ell = 0}^k F_\ell) \nabla f^c (\xs_t).
\end{align*}
Consequently, we have 
\begin{align} \label{eq:xi_intermediate}
    \xi_t = -\eta \sum\nolimits_{k=0}^K y_k = \underbrace{-\eta \sum\nolimits_{k=0}^K\! \big( \prod\nolimits_{\ell = 0}^{k-1} F_\ell)\!\!}_{:=G_t} \nabla f^c (\xs_t). 
\end{align}
Now, recall that $\xs_t =  \xc_t + \zeta_t$ due to \eqref{eq:zeta_t}.
Furthermore, if  $t'$ is the most recent round before $t$ when the client was selected, then $\xc_t = \xc_{t'+1} = \xs_{t'} + \xi_{t'}$. 
Thus, from \eqref{eq:xi_intermediate}
\begin{align*}
    \xi_t & = G_t  \nabla f^c (\xs_t) = G_t \nabla f^c (\xs_{t'} + \xi_{t'} + \zeta_t) \\
    & \approx G_t \big( \nabla f^c (\xs_{t'}) + (\nabla^2f^c(\xs_{t'})) (\xi_{t'} + \zeta_t) \big) \\
    & = G_t (G_{t'}^{-1} + \nabla^2f^c(\xs_{t'})) \xi_{t'} + G_t \nabla^2f^c(\xs_{t'}) \zeta_t,
\end{align*}
where the last equality follows from the fact that $\xi_{t'} = G_{t'} \nabla \!f^c (\xs_{t'})$; see \eqref{eq:xi_intermediate}.\footnote{
Invertibility of $G_{t'}$ is ensured by picking a learning rate $\eta < \frac{1}{\sup_{x} \lambda_{\max}(\nabla^2f^c(x))}$. This ensures $F_\ell \succ 0$ for all $\ell$ in \eqref{eq:xi_intermediate}. 
% \TODO{Kushal: Can you massage this? E.g., Add special cases when $f^c$ is convex?}
} 
In summary, the analysis shows that 
\begin{align} \label{eq:xi_dyn}
    \xi_t \approx A_{t,t'} \xi_{t'} + B_{t,t'} \zeta_t,
\end{align}
where $A_{t,t'} = G_t (G_{t'}^{-1} + \nabla^2f^c(\xs_{t'}))$ and $B_{t,t'} = G_t \nabla^2f^c(\xs_{t'}) $ depend on the server states $\xs_t$, $\xs_{t'}$ and the gradients and Hessians of the client objective function $f^c$ evaluated at those server states. 
In the event of a failed interception, the adversary uses \eqref{eq:xi_dyn} to estimate $\xi_t$ from its knowledge/estimate of $\xi_{t'}$. 

As discussed earlier in \Cref{rem:zeta}, estimating $\zeta_t$ is highly challenging task. 
Similarly, estimation of $A_{t.t'}$ requires knowledge of the server states and the gradients and Hessians of $f^c$. 
To mimic \eqref{eq:xi_dyn} in the absence of such knowledge the adversary may follow 
\begin{align} \label{eq:xihat_dynamics}
    &\xihat_t = \begin{cases}
        M \xihat_{t'},  \qquad &\mu_t = 0,\\
        \xi_t, &  \mu_t = 1,
    \end{cases}~ %; \qquad\xihat_0 \text{  is chosen arbitrarily,} 
\end{align}
for a suitably chosen matrix $M$ that acts a proxy to $A_{t,t'}$.

% For the simplicity of the exposition, we assume that the client performs only one step of GD; see Appendix~\cref{AP:multi-step-GD}\TODO{} for the $K$-step GD. Now, from the definition of $\xi_t$ in \eqref{eq:xi_t}, we obtain 
% \begin{align} \label{eq:xi_dyn_pre}
%     \xi_t = - \eta \nabla f^c(\xs_t) = - \eta  \nabla f^c(\xc_t + \zeta_t),
% \end{align}
% where we have used the definition of $\zeta_t$ from \eqref{eq:zeta_t}. 
% Now, notice that during the interval $[t'+1, t]$ $\xc$ does not change since the agent was not sampled. 
% Consequently, $\xc_t = \xc_{t' + 1}$. 
% Furthermore, using the client dynamics \eqref{eq:client_dyn}, we may write $\xc_{t' + 1} = \xs_{t'} + \xi_{t'}$. 
% Therefore, \eqref{eq:xi_dyn_pre} yields,
% \begin{align*}
%     \xi_t & = - \eta  \nabla f^c(\xs_{t'} + \xi_{t'} + \zeta_t)\\
%     & \approx -\eta \big( \nabla f^c (\xs_{t'}) + (\nabla^2f^c(\xs_{t'})) (\xi_{t'} + \zeta_t) \big) \\
%     & \overset{(\dagger)}{=}  (I - \eta \nabla^2f^c(\xs_{t'})) \xi_{t'} - \eta \nabla^2f^c(\xs_{t'}) \zeta_{t}, \\
%     & = G_1 \xi_{t'} + G_2 \zeta_{t'} 
% \end{align*}
% where $(\dagger)$ follows from the definition $\xi_{t'} = -\eta \nabla f^c (\xs_{t'})$.
%% At this point, if we can show $\zeta_t \approx L \nabla f (\xs_{t'})$, we will prove the claim $\xi_t \approx G \xi_{t'} $ with $G = (I - \eta \nabla^2f^c(\xs_{t'}) +  \nabla^2f^c(\xs_{t'}) L)$.

Note that, the dynamics \eqref{eq:xihat_dynamics} is missing an initial condition. 
So far, $\xihat_t$ has been defined in terms of $\xihat_{t'}$, where $t'$ is the most recent time the client was sampled before time $t$. 
If $t_0$ is the first time the client was sample, there would be no $t_0'$. 
Therefore, to compute $\xihat_{t_0}$, we use the initialization $\xihat_{t_0} = \mu_{t_0}\xi_{t_0}$. 
That is, if the adversary can successfully eavesdrop, it uses that eavesdropped signal, which results in a zero initialization error. 
Otherwise, the adversary initializes it to be $0$. 
% Combining the dynamics \eqref{eq:xihat_dynamics} and the initial condition, we may write in a compact form 
% \begin{align*}
%     \xihat_t 
% \end{align*}

% \begin{remark}
%     In the (unlikely) event of $\mu_t = 0$ for all $t$, the estimate $\xihat_t = 0$ for all $t$. On the other hand, if $\mu_t = 1$ for all $t$, we have $\xihat_t = \xi_t$ for all $t$. \TODO{complete.}
% \end{remark}

\section{Protection Analysis} 
\label{sec:prot_analysis}

To analyze protection, we define the error $\ec_t = \xc_t - \xhat_t$ and obtain its dynamics: 
\begin{align} \label{eq:ec_dynamics}
    \ec_{t+1} &= \xc_{t+1} - \xhat_{t+1}  \overset{ \eqref{eq:client_dyn},\eqref{eqn:xhat_def} }{=}  \ec_t + \delta_t (\xi_t + \zeta_t - \za_t) \\
    &\overset{\eqref{eqn:xihat_def}}{=} \ec_t + \delta_t ( \zeta_t - \zetahat_t + (1-\mu_t) (\xi_t -\xihat_t)  ).
\end{align}

Before proceeding further, we notice that the dynamics of $\xihat_t$ in \eqref{eq:xihat_dynamics} depends on the most recent client selection time $t'$, which is a random variable. 
In other words, the difference $t-t'$ is not only time-varying, but also stochastic. 
To deal with such intricacies, let us formally define the random variable $\tau_t$ to be the most recent client sampling time up to but not including time $t$. 
Therefore, $\tau_t \le (t-1)$ almost surely. 
Furthermore, we use the notation $\tau_t = -1$ to denote that the client has not been sampled up to time $(t-1)$. 
More compactly, we may write 
\begin{align} \label{eq:tau_dynamics}
    \tau_t = \delta_{t-1} (t-1) + (1-\delta_{t-1}) \tau_{t-1}, \qquad \tau_{-1} = -1,
\end{align}
which ensures $\tau_t = t-1$ iff the client was selected at time $t-1$, and $\tau_t = -1$ iff the client has not been sampled at all up to time $t-1$. 

% Now, let us define one more quantity $w_t$ as follows
% \begin{align*}
%     w_t = \delta_t (\xi_t - \xi_{\tau_t}) + (1-\delta_t) w_{t-1}, \qquad w_{-1} = \xi_{-1} = 0.
% \end{align*}
% The quantity, $w_t$ tracks the difference in the $\xi$ variable computed by the client. 
% $w_t$ is a piece-wise constant variable that holds its value during the interval the client is not selected, and it records the temporal difference in $\xi_t$ every time the client is sampled. 

To analyze the protection, we first leverage the newly defined variable $\tau_t$ to rewrite the $\xihat_t$ dynamics \eqref{eq:xihat_dynamics}  as follows:
\begin{align} \label{eq:xihat_compact_dynamics}
    \xihat_t = \mu_t \xi_t + (1-\mu_t)M\xihat_{\tau_t}, \qquad \xihat_{-1} = 0.
\end{align}
Finally, we define an additional variable $\qq_t = \xi_{\tau_t} - \xihat_{\tau_t}$ that follows the dynamics 
\begin{align*}
    \qq_{t+1} & = \xi_{\tau_{t+1}} - \xihat_{\tau_{t+1}} \\
     & \overset{\eqref{eq:tau_dynamics}}{=} \delta_t (\xi_{t} - \xihat_{t}) + (1-\delta_t) (\xi_{\tau_{t}}- \xihat_{\tau_{t}}) \\
     & \overset{\eqref{eq:xihat_compact_dynamics}}{=} \big(\! (1-\delta_t)I\! + \delta_t(1-\mu_t)M \big)q_t + \delta_t(1-\mu_t)(\xi_t - M \xi_{\tau_t}).
\end{align*}

Now, we turn our attention back to the dynamics of $\ec_t$ in \eqref{eq:ec_dynamics} and use \eqref{eq:xihat_compact_dynamics} to rewrite it as 
\begin{align*}
    \ec_{t+1} = \ec_t + \delta_t(\zeta_t -\zetahat_t) + \delta_t (1-\mu_t)(Mq_t + \xi_t - M \xi_{\tau_t})
\end{align*}
We define the joint state as $\sigma_t = [{\ec_t}^\intercal, q_t^\intercal]^\intercal$, which follows the dynamics 
\begin{align*}
    &\sigma_{t+1} = A_t \sigma_t + u_t, \\
    \text{where} \qquad \quad &A_t =  \begin{bmatrix}
        I & \delta_t(1-\mu_t)M\\
        0 & (1-\delta_t)I + \delta_t(1-\mu_t)M 
    \end{bmatrix}, \quad\\
    &u_t = \begin{bmatrix}
       \delta_t (\zeta_t - \zetahat_t) + \delta_t(1-\mu_t)(\xi_t - M\xi_{\tau_t}) \\
       \delta_t(1-\mu_t)(\xi_t - M\xi_{\tau_t})
    \end{bmatrix}.
\end{align*}

By defining the quantity $\Sigma_t = \E[\sigma_t \sigma_t^\intercal]$, we notice that the protection at time $t$ can be found from $\tr([I ~~0] \Sigma_t [I ~~0]^\intercal)$ and likewise, the final protection is $\tr([I ~~0] \lim_{t\to \infty} \Sigma_t [I ~~0]^\intercal)$. Here the expectation is taken with respect to randomness of $\{\delta_t\}_{t\ge 0}$ and $\{\mu_t\}_{t\ge 0}$.
Therefore, to analyze the protection, we investigate the time-evolution of $\Sigma_t$:
\begin{align*} %\label{eq:Sigma_dynamics}
\begin{split}
    \Sigma_{t+1} = \E[A_t \Sigma_t& A_t^\intercal] + \E[u_t u_t^\intercal] % \\ &
    + \E[A_t \sigma_t u_t^\intercal] + \E[u_t \sigma_t^\intercal A_t^\intercal],
\end{split}
\end{align*} 
where we have used $\E[A_t \sigma_t \sigma_t^\intercal A_t^\intercal ] = \E[\E[A_t \sigma_t \sigma_t^\intercal A_t^\intercal |  $ $\mu_t,  \delta_t]] = \E[A_t \E[\sigma_t \sigma_t^\intercal | \mu_t, \delta_t] A_t^\intercal] = \E[A_t \Sigma_t A_t^\intercal]$. 

We further notice that the adversary must ensure that the $\Sigma_t$ dynamics is stable, otherwise, the client may achieve an infinite amount of protection.
We will next show that the stability of the $\Sigma_t$ dynamics is solely controlled by the adversary by picking an appropriate $M$. 
The following lemma provides the necessary  condition on $M$ to ensure stability of $\Sigma_t$.

\begin{lemma} \label{lem:M_condition}
    A necessary condition for the stability of $\Sigma_t$ is that all the eigenvalues of $M$ must have a magnitude less than $(p(1-\gamma)\max\{\gamma, 1-\gamma\})^{-\frac{1}{2}}$. %\TODO{Define $p$}.
\end{lemma}

\begin{proof}
    The stability of the $\Sigma_t$ dynamics is solely determined by the term $\E[A_t \Sigma_t A_t^\intercal]$. One may verify (see \Cref{AP:someUsefulResults}) that $\E[A_t \Sigma_t A_t^\intercal]  = \mathcal{L} (\Sigma_t)$, with the linear operator $\mathcal{L}$ being defined as 
    \begin{align} \label{eq:operator:L}
       \mathcal{L}\! (\Sigma_t) =\!   (1\!-\!p) \Sigma_t + p K_1\!\Sigma_t K_1^\intercal +p\gamma(1\!-\!\gamma) K_2\! \Sigma_t K_2^\intercal, 
    \end{align}
    \begin{align*} 
     \text{where, }\qquad K_1 = \begin{bmatrix}
            I & (1-\gamma)M\\
            0 & (1-\gamma)M
        \end{bmatrix},\quad
    K_2 =  \begin{bmatrix}
            0 & M\\
            0 & M
        \end{bmatrix}. \hspace{1cm}
    \end{align*}
    The necessary and sufficient condition for stability is to ensure that the spectral radius of $\mathcal{L}$ is less than 1. 
    Since the spectral radius of $\mathcal{L}$ is  greater than or equal to $\max\{p\lambda_{\max}^2(K_1),~~ p\gamma(1-\gamma)\lambda_{\max}^2(K_2)\}$, 
    the necessary condition becomes
    \begin{align} \label{eq:Pre_necessary_condition}
        1 > \max\{p\lambda_{\max}^2(K_1),~~ p\gamma(1-\gamma)\lambda_{\max}^2(K_2)\}.
    \end{align}
    Further investigation on $K_1$ and $K_2$ (see \eqref{eq:K1_K2} in \Cref{AP:someUsefulResults})  reveals $\lambda_{\max}(K_1) = \max\{1, (1-\gamma)\lambda_{\max}(M)\}$ and $\lambda_{\max}(K_2) = \lambda_{\max}(M)$.
    Therefore, the necessary condition \eqref{eq:Pre_necessary_condition} simplifies to 
    \begin{align*}
        1 > p (1-\gamma) \max\{\gamma, 1-\gamma\} \lambda_{\max}^2(M).
    \end{align*}
    This proves the necessary condition.
\end{proof}

\begin{remark}
As $p$ approaches zero and/or $\gamma$ approaches 1, the allowable upper bound on $\lambda_{\max}(M)$ increases, thereby enlarging the set of feasible matrices that satisfy the condition in \Cref{lem:M_condition}.
When $\gamma$ approaches 1, the need to estimate $\xi$ diminishes, as the adversary is able to intercept messages more frequently.
Consequently, the estimation dynamics of $\xihat$—and thus the role of $M$—become less significant. Likewise, when $p\to 0$, $M$ barely affects the error dynamics since $\bbP(\ec_{t+1} = \ec_t) \ge 1-p$.\footnote{
From \eqref{eq:ec_dynamics} we observe that the event $\{\delta_t=0\}$ implies $\ec_{t+1} =  \ec_t$. Thus, $\bbP(\ec_{t+1} =  \ec_t) \ge  \bbP(\delta_t = 0) = 1-p.$
} \hfill $\blacksquare$
\end{remark}

Without loss of any generality, for the subsequent analysis we assume that the adversary picks an $M$ satisfying the condition in \Cref{lem:M_condition}. 
To facilitate the computation of the protection, let us denote the quantities $\bar\sigma_t := \E[\sigma_t]$, $\tilde{\sigma}_t := \sigma_t - \bar\sigma_t$, $\bar A := \E[A_t]$, $\tilde A_t := A_t - \bar A$, $\bar u_t := \E[u_t]$, and $\tilde u_t := u_t - \bar u_t$. 
Then, 
\begin{align}
    &\bar\sigma_{t+1} = \bar A \bar \sigma_t + \bar u_t, \label{eq:sigma_bar}\\
    & \tilde \sigma_{t+1} = A_t \tilde \sigma_t + \underbrace{\tilde A_t \bar \sigma_t + \tilde u_t}_{:= v_t}. \label{eq:sigma_tilde}
\end{align}
Using the definitions of $\bar\sigma_t$ and $\tilde\sigma_t$, we may write $\Sigma_t = \E[\tilde\sigma_t \tilde \sigma_t^\intercal] + \bar\sigma_t \bar\sigma_t^\intercal$. 
By defining $\tilde\Sigma_t = \E[\tilde\sigma_t \tilde \sigma_t^\intercal]$, we obtain its dynamics
\begin{align} \label{eq:Sigma_tilde_Dynamics}
    \tilde\Sigma_{t+1} = \E[A_t \tilde\Sigma_t A_t^\intercal] + \E[v_t v_t^\intercal], 
\end{align}
where we have used $\E[A_t \tilde\sigma_t v_t^\intercal] = 0$.\footnote{The proof is left as an exercise to the reader.}  
At this point, we are ready to state the main result on the amount of protection.

\begin{theorem} \label{thm:main}
    The final protection is lower bounded as follows
    \begin{align*}
        \liminf_{t\to\infty} \|\ec_t\|^2 \ge \liminf_{t\to \infty} \{ (1-p)\|s_t + r_t\|^2 + p\gamma\|s_t\|^2\\ + (1-p)\frac{\gamma}{1-\gamma} \|r_t\|^2 +  \E[\| \xi_{\tau_t}-\E[\xi_{\tau_t}]\|^2]
        +\|g_t\|^2\},
    \end{align*}
    where
    \begin{align*}
        &r_t = \zeta_t-\zetahat_t, \\
        &s_t = h_t + M\E[q_t],\\
        &h_t = (1-p)^t\xi_t + \sum\nolimits_{k=0}^{t-1}(1-p)^{t-k-1}p(\xi_t - M\xi_k).  \end{align*}  \begin{align*}
        & g_t = \xc_o - \xhat_0 + p\big(\sum\nolimits_{k=0}^{t} r_k\big)\\
        &\qquad  + p(1-\gamma)(I-(1-\gamma)M)^{-1}\sum\nolimits_{k=0}^{t}h_k.
    \end{align*}
    % $r_t = \zeta_t-\zetahat_t $, $s_t = \xi_t - M\E[\xi_{\tau_t}] + M\E[q_t]$, and  $g_t = (\xc_o - \xhat_0 + p\sum_{k=0}^{t} r_k + p(1-\gamma)(I-(1-\gamma)M)^{-1}\sum_{k=0}^{t}(\xi_k - M\E[\xi_{\tau_k}]))$
\end{theorem}
\begin{proof}
    Recall that,
    \begin{align*}
        \|\ec_t\|^2 &= \tr([I ~~0] \Sigma_t [I~~0]^\intercal ) \\
        &= \tr([I ~~0] \tilde \Sigma_t [I~~0]^\intercal ) + \tr([I ~~0] \bar \sigma_t \bar\sigma_t^\intercal [I~~0]^\intercal ).
    \end{align*}
    Let us focus on the term $ \tilde\Sigma_{11,t} : =  [I ~~0] \tilde \Sigma_t [I~~0]^\intercal$. 
    To that end, we define $V_t \triangleq [I ~~0] \E[v_t v_t^\intercal] [I ~~0]^\intercal $ and pre and post multiply \eqref{eq:Sigma_tilde_Dynamics} with $[I~~0]$ and $[I~~0]^\intercal$, respectively, to obtain
    \begin{align*}
        \tilde\Sigma_{11, t+1} & = [I ~~0] \E[A_t \tilde\Sigma_t A_t^\intercal]  [I~~0]^\intercal + V_t \\
        & \overset{(\dagger)}{=} [I ~~0] \mathcal{L}(\tilde\Sigma_t)[I~~0]^\intercal + V_t \\
         = (1&-p + p\gamma) \tilde\Sigma_{11, t} + p(1-\gamma)[I ~~M]\tilde{\Sigma}_t [I~~M]^\intercal + V_t,
    \end{align*}
    where we used \Cref{prop:L} for $(\dagger)$ and the expression of $\mathcal{L}(\cdot)$ from \eqref{eq:operator:L} for the last equality. 
    Therefore, we obtain 
    \begin{align*}
        \tilde\Sigma_{11, t+1} \succeq (1&-p + p\gamma) \tilde\Sigma_{11, t} + V_t.
    \end{align*}
    % and hence,  $\lim_{t\to \infty} \tilde\Sigma_{11, t} \succeq \frac{1}{p(1-\gamma)} \lim_{t\to \infty} V_t$ when $p(1-\gamma)> 0$. 
    and hence,
    \begin{align}
        \tr(\tilde\Sigma_{11, t+1}) &\ge (1-p+p\gamma)\tr(\tilde\Sigma_{11, t}) + \tr(V_t) \nonumber \\
        &\ge \sum_{k=0}^t(1-p+p\gamma)^{t-k} \tr(V_k),
    \end{align}
    and 
    \begin{align*}
        \liminf_{t\to \infty} tr(\tilde\Sigma_{11, t+1}) \ge \frac{1}{p(1-\gamma)} \liminf_{t\to\infty} \tr(V_t).
    \end{align*}
     Now, one may verify that\footnote{See \Cref{prop:Vt} in \Cref{AP:someUsefulResults}.} 
    \begin{align*}
        \tr(V_t) = p(1-p)(1-\gamma)\|s_t + r_t\|^2 + p^2\gamma(1-\gamma)\|s_t\|^2\\ + p(1-p)\gamma \|r_t\|^2 + p(1-\gamma) \E[\| \xi_{\tau_t}-\E[\xi_{\tau_t}]\|^2].
    \end{align*}
    % \begin{align*}
    %     V_t = p(1-p)(1-\gamma)(s_t+r_t)(s_t+r_t)^\intercal + p^2\gamma(1-\gamma)s_ts_t^\intercal \\
    %     +p(1-p)\gamma r_t r_t^\intercal + p(1-\gamma)\E[(\E[\xi_{\tau_t}]- \xi_{\tau_t})(\E[\xi_{\tau_t}]- \xi_{\tau_t})^\intercal]. 
    % \end{align*}
    % \begin{align*}
    %     \frac{1}{p(1-\gamma)} \lim_{t\to \infty} V_t = \lim_{t\to \infty} \big( (1-p)(s_t+r_t)(s_t+r_t)^\intercal + p\gamma s_ts_t^\intercal \big) 
    % \end{align*}
    To complete the proof, we verify from \eqref{eq:sigma_bar} that 
    \begin{align*}
        [I~~0]&\bar\sigma_t = (\xc_0- \xhat_0) + p\!\sum_{k=0}^{t-1}r_k \\
        &+ p(1-\gamma)(I-M_2)^{-1} \!\!\sum\nolimits_{k=0}^{t-1}(pI \! - \! M_1 M_2^{t-1-k})h_k ,
    \end{align*}
    where $M_1 = p(1-\gamma)M$ and $M_2 = (1-p)I + M_1$.
    Finally, this leads to 
    \begin{align*}
        \liminf_{t\to \infty} [I ~~0] \bar\sigma_t \bar\sigma_t^\intercal\begin{bmatrix}
            I\\0
        \end{bmatrix} = \liminf_{t\to\infty}\|g_t\|^2, 
    \end{align*}
     where we have used that, for any finite $k$, $\lim_{t\to\infty} M_2^{t-1-k}\!=\!0 $  due to \Cref{lem:M_condition} ensuring $\lambda_{\max}(M_2) < 1$.
    Thus, 
    \begin{align*}
        \liminf_{t\to\infty} \E\|\ec_t\|^2 \ge \liminf_{t\to\infty}\left( \frac{1}{p(1-\gamma)}\tr(V_t) + \|g_t\|^2 \right),
    \end{align*}
    which completes the proof.
\end{proof}
\Cref{thm:main} illustrates how the different aspects of the problem influences the protection. 
For instance, the roles of $p$ and $\gamma$ are quite obvious. 
More importantly, we also observe how $r_t = \zeta_t - \zetahat_t$ affects the protection. Similarly, the role of $\xi_t$ in the protection is also explicit. 
At a first glance, the protection lower bound in \Cref{thm:main} may not provide immediate insights on these roles. Therefore, we discuss a special case in details.
However, before that, we note that the quantity $\E[\| \xi_{\tau_t}-\E[\xi_{\tau_t}]\|^2]$ in \Cref{thm:main} can be simplified as
\begin{align*}
    \E[\| \xi_{\tau_t}-\E[\xi_{\tau_t}]\|^2] = & (1-p)^t\|\ell_t\|^2 \\
    &+ \sum\nolimits_{k=0}^{t-1}p(1-p)^{t-k-1}\|\xi_k - \ell_t\|^2,\\
    \ell_t := & \E[\xi_{\tau_t}] = \sum\nolimits_{s=0}^{t-1}p(1-p)^{t-s-1}  \xi_s,
\end{align*}
 which further shows the influence of $p$ and $\xi$ on the protection.

\subsubsection{Perfect Eavesdropping Case: $\mu_t =1$ almost surely}
\label{sub:mut_1}
% \kc{comment on $\mu_t = 1 \forall t$}

We revisit the $\tilde\Sigma_t$ dynamics \eqref{eq:Sigma_tilde_Dynamics} and substitute $\mu_t =1$, which yields
\begin{align*}
    \tilde \Sigma_{t+1} = &(1-p)\tilde\Sigma_t + p \begin{bmatrix}
        I & 0\\
        0 & 0
    \end{bmatrix}\tilde \Sigma_t \begin{bmatrix}
        I & 0\\
        0 & 0
    \end{bmatrix} \\
    &+ \E[(\delta_t-p)^2] \begin{bmatrix}
        r_t \\
        -\E[q_t]
    \end{bmatrix}\begin{bmatrix}
        r_t \\
        -\E[q_t]
    \end{bmatrix}^\intercal.
\end{align*}
 Therefore, we have 
 \begin{align*}
     \tr(\tilde\Sigma_{11, t+1}) = \tr(\tilde\Sigma_{11,t}) + p(1-p)\|r_t\|^2 = p(1-p) \sum_{k=0}^t\|r_k\|^2. 
 \end{align*}
 Likewise, we also obtain that 
 \begin{align*}
     [I~~0]\bar\sigma_t = (\xc_0- \xhat_0) + p\sum_{k=0}^{t-1}r_k.
 \end{align*}
 In this case, the protection becomes 
 \begin{align}
     \liminf_{t\to\infty} \E\|\ec_t\|^2 = p(1-p)\liminf_{t\to\infty} \sum\nolimits_{k=0}^t \|r_k\|^2 \nonumber\\
     + \liminf_{t\to \infty} \|(\xc_0- \xhat_0) + p\sum\nolimits_{k=0}^{t-1}r_k\|^2. \label{eqn:prot_mu1}
 \end{align}
 In this case, only $\zeta_t$, i.e., the mismatch between the server and client models,  and $\xc_0-\xhat_0$ influence the protection. $\xi_t$ does not affect the protection, as one would expect. 
 It is noteworthy that setting $\xhat_0 = \xc_0$ is not necessarily optimal for the adversary. 
 Furthermore, the effect of $p$ is quite interesting. 
 % Assuming $\xc_0 = \xhat_0=0$, the protection results into $p(1-p)\liminf_{t\to\infty}\sum_{k=0}$
  By changing the sampling probability $p$, the server may improve the protection of the client. 
 In fact, \eqref{eqn:prot_mu1} is quadratic in $p$, which allows us to easily compute the optimal $p^*$ to maximize the r.h.s. of \eqref{eqn:prot_mu1}. 
 
We conclude this section by emphasizing a key insight:~the simple fact that the client shares model increments---rather than the full model---provides inherent protection, even when the adversary can eavesdrop on all communications.
In stark contrast, we next show the other extreme: FLOP-type algorithms have zero protection even when the adversary’s eavesdropping probability approaches zero.

%  \subsection{Client is Selected at Every Round: $\delta_t=1$}
%  In this case $p=1$ and $\tau_t = (t-1)$. 
%  Directly from \Cref{thm:main} we obtain 
%  \begin{align*}
%      \liminf_{t\to\infty} \E\|\ec_t\|^2  \ge \liminf_{t\to\infty}\gamma \|s'_t\|^2 + \|g'_t\|^2, 
%  \end{align*}
%  where $s'_t = \xi_t -M\xi_{t-1} + M\sum_{k=0}^{t-1}(1-\gamma)^{t-k}M^{t-1-k}(\xi_k -M\xi_{k-1})$ $g'_t = \xc_o - \xhat_0 + \big(\sum\nolimits_{k=0}^{t} r_k\big)
%          + (1-\gamma)(I-(1-\gamma)M)^{-1}\sum\nolimits_{k=0}^{t}(\xi_k - M\xi_{k-1})$
% \dm{anything interesting here?}

\subsection{FL with Clients Sending Local Models (FLOP)}
In this section, we shift our focus to FL algorithms, such as MOCHA, FL via STC, SCAFFOLD, FSVRG, where clients upload their local models to the server.  
In this case, the eavesdropper intercepts the model $\xc_t$ with probability $\gamma$. 
We let the adversary follow the exact same dynamics as \eqref{eq:adv_dynamics}, except since the adversary intercepts $\xc_t$ instead of $\xi_t$, we change the definition of $\za_t$ as follows
\begin{align*}
    \za_t = (\mu_t(\xc_{t+1} - \xhat_t) + (1-\mu_t) \xihat_t) + \zetahat_t,
\end{align*}
where $\xc_{t+1} - \xhat_t$ acts a proxy to $\xi_t$. 
Plugging this expression of $\za_t$ in \eqref{eqn:xhat_def}, we obtain 
\begin{align*}
    \xhat_{t+1} = \mu_t\delta_t \xc_{t+1} + (1- \delta_t \mu_t) \xhat_t + \delta_t(1-\mu_t) \xihat_t + \delta_t\zetahat_t.
\end{align*}
Following our earlier discussion that $\zetahat_t = 0$, we notice that the event ${\delta_t\mu_t = 1}$ implies $\xhat_{t+1} = \xc_{t+1}$. 
Therefore, if there exists a time $t$ large enough when $\delta_t = \mu_t = 1$---which, in fact, occurs with probability 1 for any $p, \gamma > 0$---the adversary has the latest model that is very close to where the clients model will converge. 
% Thus, being able to intercept at a late enough round helps the adversary 
% Consequently, for any $T$
% \begin{align*}
%     \bbP(\exists t \ge T : \xhat_{t}  = \xc_{t}) & \ge \bbP(\cup_{t\ge T} ~\{\delta_t\mu_t =1\}) \\
%     & = 1 - \prod\nolimits_{t\ge T} \bbP(\mu_t\delta_t \ne 1) = 1.
% \end{align*}
% for all $p, \gamma >0$. 
This loss of privacy due to clients sending the true models instead of the model increments has been observed in other similar topics such as remote estimation \cite{tsiamis2019state}, consensus protocols \cite{arXivMaity2021}, and distributed control \cite{maity2024ensuring}. We emphasize that $\xhat_{t+1} = \xc_{t+1}$ occurs with probability $1$ for sufficiently large $t$, even when $\gamma$ is arbitrarily close to zero. 
In other words, FLOP-type algorithms are $0$\textit{-protected}, regardless of how weak the adversary is,  as long as $\gamma>0$.

\section{Experimental Results}
\label{sec:exp}

Here, we present experimental results demonstrating the protection-efficacy of FLIP algorithms, where the clients upload only the model updates $\xi_t$ to the server, against adversary's eavesdropping mechanism~\eqref{eqn:xhat_def},\eqref{eqn:xihat_def},\eqref{eq:xihat_dynamics} with eavesdropping success probability $\gamma = 0.5$. For simplicity, we choose $M = 0.5 I$. In general, the matrix-vector multiplication $M \xihat_{t'}$ in~\eqref{eq:xihat_dynamics} would require the adversary to do $\bigO(d^2)$ multiplications whenever $\mu_t = 0$, where $d$ denotes the dimensionality of the model. Instead, the choice $M = m I$ for some real-valued scalar $m > 0$ results in the adversary dynamics computation of $\bigO(d)$ for each $t$, reducing the adversary's computational budget, especially for large models. Following Remark~\ref{rem:zeta}, we also set $\zetahat_t = 0$.

We conduct two sets of experiments. Both experiments are for solving image classification problems by training the LeNet-5~\cite{lecun1998gradient} convolutional neural network (CNN) model with CIFAR-10 dataset~\cite{krizhevsky2009learning}. The benchmark CIFAR-10 dataset consists of $60k$ $32 \times 32$ color images in $10$ classes, with $6000$ images per class. There are $50k$ training images and $1k$ test images. LeNet-5 is one of the earliest CNN models that introduced key concepts such as convolution, pooling, and hierarchical feature extraction that underpin modern deep learning models.

In the first experiment, we are interested in evaluating the adversary's estimates $\xhat_t$. We ask the question: {\it how can the adversary judge the quality of its estimates $\xhat_t$?} Since the adversary does not have access to the client's actual models $\xc_t$, a meaningful way to judge $\xhat_t$'s quality would be evaluating it on representative samples from the same distribution as the client's data. Such samples can be obtained from the CIFAR-10 test set. There are situations where the adversary shares a common knowledge of the clients' data distribution or has access to a similar data pool, or in some cases, even the same underlying dataset~\cite{nasr2019comprehensive, lyu2020threats, fang2024byzantine}. So, we assume the client can access the CIFAR-10 test set. Note that, we remove this assumption in the second experiment later. We will evaluate and compare the adversary's $\xhat_t$ and the client's $\xc_t$ on the CIFAR-10 test set. Without loss of generality, we assume that the adversary is monitoring the uplink of the first client.

The FLIP training is conducted as follows. The CIFAR-10 training set is equally partitioned horizontally~\cite{fu2024differentially} among $N=8$ clients, so each client has $6.25k$ training samples. Since the performance of a federated learning algorithm depends on several factors, such as clients selection, batch-size of the clients, and number of local epochs $K$, we consider four scenarios. First, we describe the base setting as follows. The server chooses $n = 5$ clients uniformly at random at each global round. Since the adversary monitors the first client's uplink, we ensure that the server selects this client at least in the last five global rounds, so we have enough instances for demonstration. The number of global rounds is $T= 30$, and the number of local epochs is $K=3$. The batch-size of each client is $B = 128$. The learning rate is $\eta = 10^{-3}$. In the second setting, we consider that all the $8$ clients participate in each global round. In the third setting, local epochs are increased to $K = 10$. In the fourth setting, batch-size is reduced to $B = 8$. In each of these settings, the rest of the parameters of federated learning remain the same as in the first setting. The global model $x^s_0$ is initialized randomly and identically for all four settings.

In each of these four settings of the FLIP algorithm, we compute the adversary estimates with $\gamma = 0.5$. We consider two initialization of $\xhat_{-1}$: one is $\xhat_{-1} = 0$ and the other is $\xhat_{-1} = \xc_{t_0}$ where $t_0$ indicates the first global round $t = t_0$ where the first client is selected by the server. The second choice of initialization is to check the quality of adversary estimates if it is initialized close to the actual client trajectory. The accuracies of $\{\xhat_t\}$ and $\{\xc_t\}$ are evaluated on the CIFAR-10 test, shown in Fig.~\ref{fig:advclient}. Both client's model and adversary's estimate do not update over those global rounds $t$ when $\delta_t = 0$. We observe that the FLIP algorithm reaches steady-state in $T = 30$ rounds, and the first client's test accuracy is stable at values within $55\%-63\%$. However, the adversary's test accuracy is $\approx10\%$ for zero initialization and $20\%-30\%$ for $\xhat_{-1} = \xc_{t_0}$. At the steady-state, adversary's test accuracy is $\approx 12\%$ (ref. Fig.~\ref{fig:advclient_8}). When $\mu_t = 1$ almost surely, we observe the adversary's steady-state test accuracy is merely $\approx 13\%$, slightly better than a Bernoulli $\mu_t$ (ref. Fig.~\ref{fig:advclient_8}), indicating a non-zero protection as proved in~\eqref{eqn:prot_mu1}. On a side note, the client's test accuracy is the least $(55\%)$ for a larger number of local epochs (ref. Fig.~\ref{fig:advclient_10}) due to the client's model drift phenomena~\cite{fu2024differentially}.

\begin{figure}
\centering
\begin{subfigure}{.245\textwidth}
  \begin{center}
  \includegraphics[width = \textwidth]{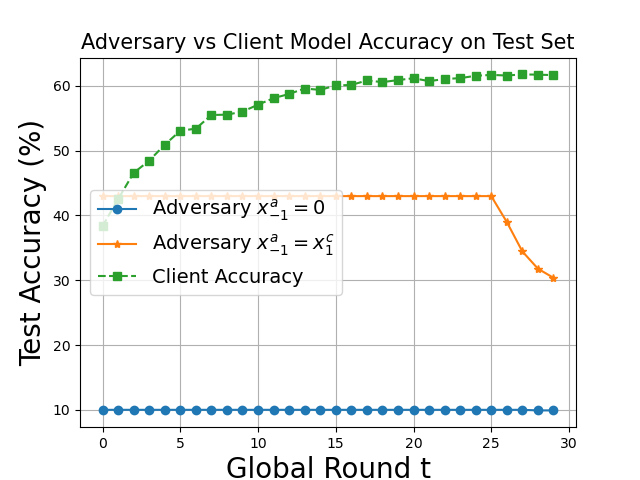}
  \caption{subset of clients participation}
  \label{fig:advclient_4}
  \end{center}
\end{subfigure}%
\begin{subfigure}{.245\textwidth}
  \begin{center}
  \includegraphics[width = \textwidth]{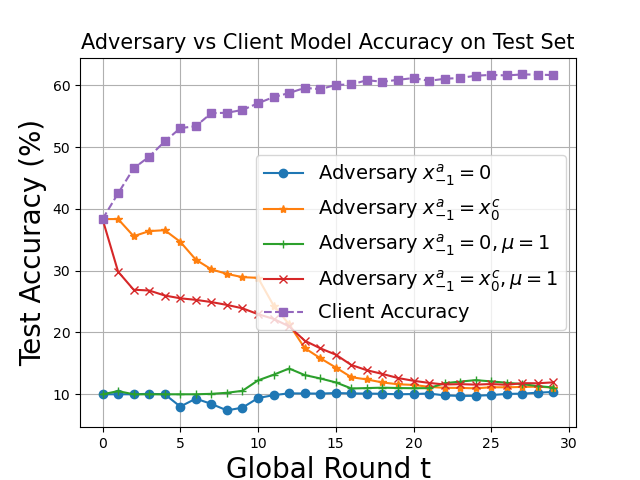}
  \caption{all clients participation}
  \label{fig:advclient_8}
  \end{center}
\end{subfigure}
\bigskip 
\begin{subfigure}{.245\textwidth}
  \begin{center}
  \includegraphics[width = \textwidth]{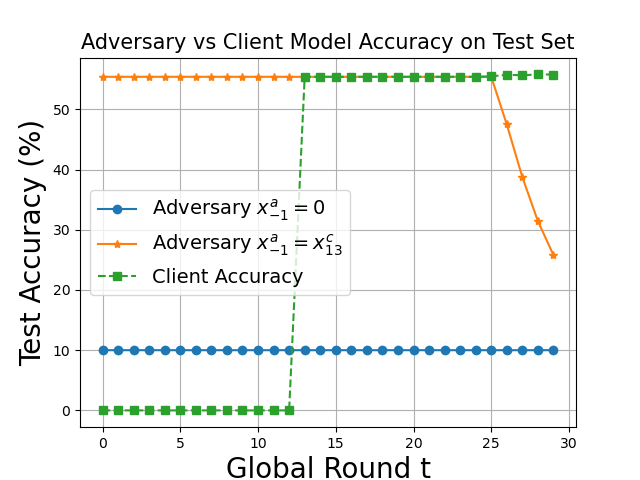}
  \caption{increased local epochs}
  \label{fig:advclient_10}
  \end{center}
\end{subfigure}%
\begin{subfigure}{.245\textwidth}
  \begin{center}
  \includegraphics[width = \textwidth]{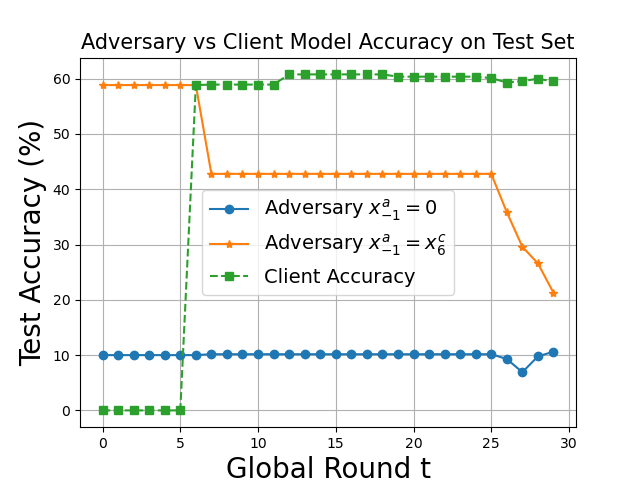}
  \caption{less batch-size}
  \label{fig:advclient_bs8}
  \end{center}
\end{subfigure}
\caption{\it \small CIFAR-10 test set accuracies, of (i) adversary's estimated model $\xhat_t$ with different initialization and (ii) client's true model $\xc_t$, when LeNet-5 is trained using FLIP. FL settings: (a) $n=5, K = 3, B = 128$, (b) $n=8, K = 3, B = 128$, (c) $n=5, K = 10, B = 128$, (d) $n=5, K = 3, B = 8$.}
\label{fig:advclient}
\end{figure}

In the second experiment, we compare FLIP's client model protection with DP-FL's client data privacy. Here, DP-FL refers to FedAvg equipped with a standard differential privacy mechanism, where the client uploads $\xc_t$ to the server (a FLOP algorithm). We remove the assumption that the adversary has access to any CIFAR-10 test samples. Since FLIP provides model protection and DP-FL provides training data protection, we require a common metric to compare these two mechanisms. We choose this metric as the adversary's estimated model accuracy on the reconstructed client's training data. An eavesdropping adversary can reconstruct a client's training data if, at any global round $t$, it has access to the accumulated gradients $\xi_t$ (or two consecutive client models $\xc_t$ and $\xc_{t-1}$), an estimate of learning rate $\eta$, and the actual or estimate of client's model $\xc_t$. Several model inversion algorithms allow the adversary to compute estimates of the client's training data~\cite{zhu2019deep, zhao2020idlg, geiping2020inverting, wilson2024federated}. Similar to the first experiment, we consider four different federated learning settings.

In each such setting, we begin with implementing a standard DP-FL algorithm (NbAFL)~\cite{wei2020federated} that incorporates differential privacy of client's data during training of federated learning. In implementing NbAFL, we set the privacy budget parameters $\delta = 0.01, \epsilon = 100$, clipping threshold of model weights $C=70$ (following the model observation process during the course of training mentioned in~\cite{wei2020federated}, and an upper limit on the number of rounds the first client participates is $10$. We observed a significant degradation of the client model's accuracy with a stricter privacy budget of less than $100$. That is why we set $\epsilon = 100$. We did not add noise in the server's global model update in NbAFL, since we assumed that the adversary only monitors the uplink.

If the clients upload its updated model $\xc_t$ as in DP-FL, instead of a model difference $\xi_t$ as in FLIP, the adversary has access to $\xc_t$ whenever its eavesdropping is successful, i.e., $\mu_t = 1$. Since the eavesdropping success probability is $\gamma > 0$, we will have $\mu_t = 1$ for some global round $t$ given a sufficient number of global rounds. Thus, for our implementation, we assume the adversary successfully eavesdrops on two consecutive global rounds $t=28$ and $t=29$. So, in the case of DP-FL, the adversary knows the client models $\xc_{28}$ and $\xc_{29}$. From this information, we let the adversary reconstruct the client's training data using the inversion algorithm from~\cite{zhu2019deep}. The model inversion algorithm is implemented for $100$ iterations, and the optimizer is chosen as L-BFGS. A sample of adversary's reconstructed images is shown in Fig.~\ref{fig:dp} for batch-size $B = 8$. On this reconstructed training data set, the adversary evaluates the intercepted true client model $\xc_{28}$. The corresponding accuracies are shown in Fig.~\ref{fig:dpis}. Since the adversary's estimate does not update over those global rounds $t$ when $\delta_t = 0$, we do not plot accuracies over those rounds.

\begin{figure}
\centering
\begin{subfigure}{.245\textwidth}
  \begin{center}
  \includegraphics[width = \textwidth]{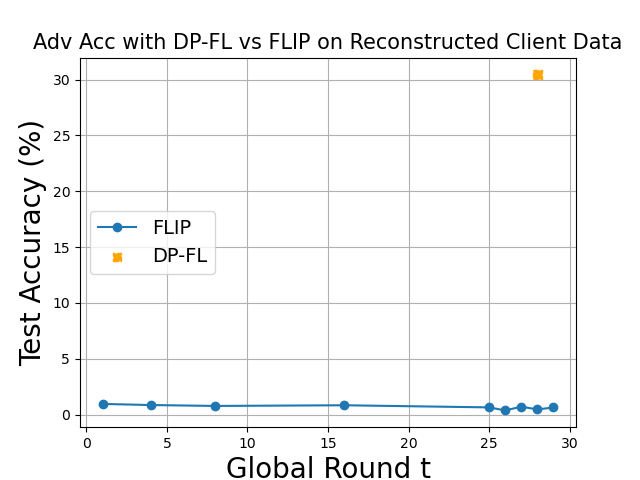}
  \caption{subset of clients participation}
  \label{fig:dpis_4}
  \end{center}
\end{subfigure}%
\begin{subfigure}{.245\textwidth}
  \begin{center}
  \includegraphics[width = \textwidth]{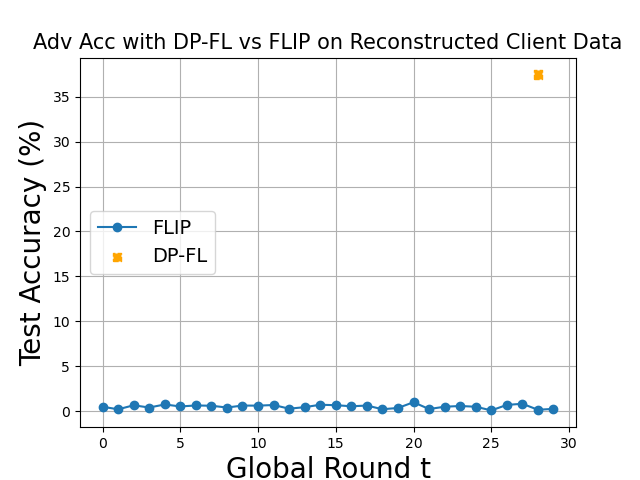}
  \caption{all clients participation}
  \label{fig:dpis_8}
  \end{center}
\end{subfigure}
\bigskip 
\begin{subfigure}{.245\textwidth}
  \begin{center}
  \includegraphics[width = \textwidth]{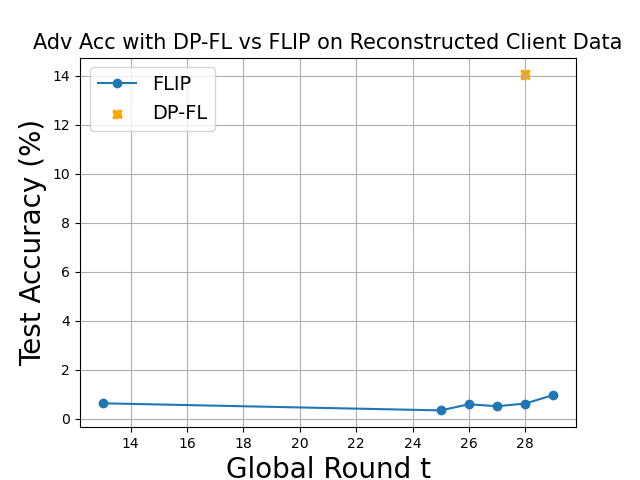}
  \caption{increased local epochs}
  \label{fig:dpis_10}
  \end{center}
\end{subfigure}%
\begin{subfigure}{.245\textwidth}
  \begin{center}
  \includegraphics[width = \textwidth]{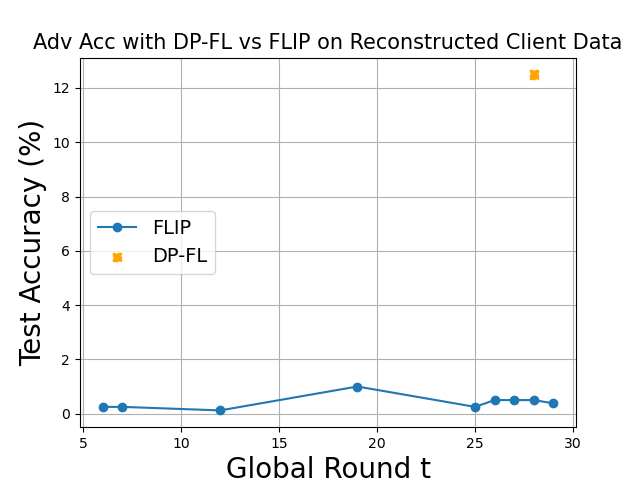}
  \caption{less batch-size}
  \label{fig:dpis_bs8}
  \end{center}
\end{subfigure}
\caption{\it \small Accuracies on reconstructed client's training images, of (i) adversary's estimated model $\xhat_t$ in case of FLIP and (ii) adversary's true intercepted model 
$\xc_t$ in case of DP-FL. FL settings are the same as Fig.~\ref{fig:advclient}.}
\label{fig:dpis}
\end{figure}

Again, in each of these four settings, we implement FLIP (without any differential privacy mechanism), as done in the first experiment. The adversary follows its estimation dynamics described in Section~\ref{sub:adv_model} to construct model estimates $\xhat_t$. From the intercepted $\xi_t$, which is the client's local gradient accumulation, adversary implements the same model inversion algorithm as above to reconstruct the client's data. A sample of adversary's reconstructed images is shown in Fig.~\ref{fig:is} for batch-size $B = 8$. Although not close to the original images, we observe that these reconstructed images in case of FLIP contain more information than DP-FL, such as some outlines are more prominent in Fig.~\ref{fig:dpis}, which is expected as differential privacy protects client data in case of DP-FL. Specifically, the \textit{mean squared error} between these two sets of images is $72.47$, indicating significant pixel-wise differences, and \textit{structural similarity index}~\cite{wang2004image} is $0.25$, meaning notable differences in structure and patterns ($0$ is completely different images). Finally, the adversary evaluates its model estimate $\xhat_t$ on this reconstructed training data set. The corresponding accuracies are shown in Fig.~\ref{fig:dpis}. We observe that the accuracy of the reconstructed data set, evaluated by the adversary with its estimated model, is $<1\%$ for FLIP. Even though the reconstructed images have better quality (less protection) for FLIP, the accuracy of reconstructed images, evaluated by the adversary with its intercepted client model, is $13\%-38\%$ for DP-FL. Moreover, regarding utility of FL, there is no loss in the client’s accuracy in federated learning with FLIP, whereas the added noising of the client models
in differential privacy mechanism reduces the client model’s accuracy on test set (ref. Fig.\ref{fig:dp_test}).

\begin{figure*}[htb!]
\centering
\begin{subfigure}{.45\textwidth}
  \begin{center}
  \includegraphics[width = 0.8\textwidth]{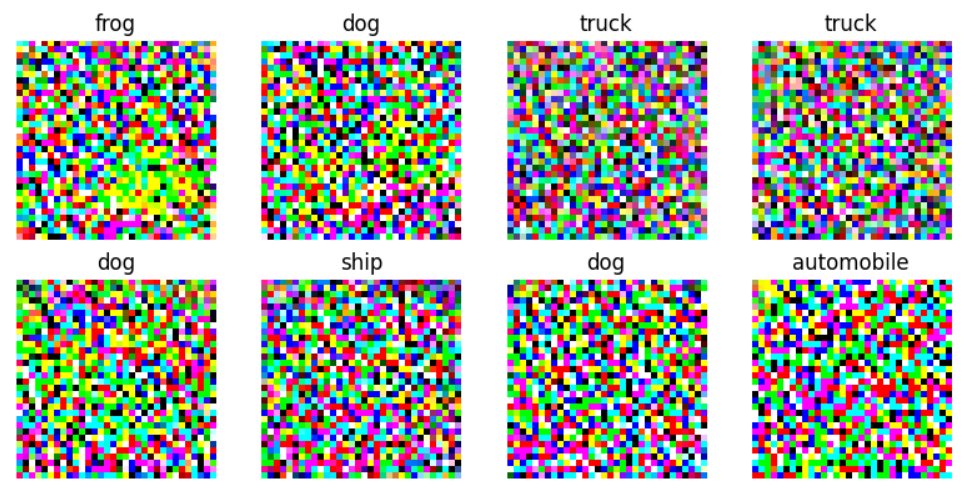}
  \caption{DP-FL}
  \label{fig:dp}
  \end{center}
\end{subfigure}%
\begin{subfigure}{.45\textwidth}
  \begin{center}
  \includegraphics[width = 0.8\textwidth]{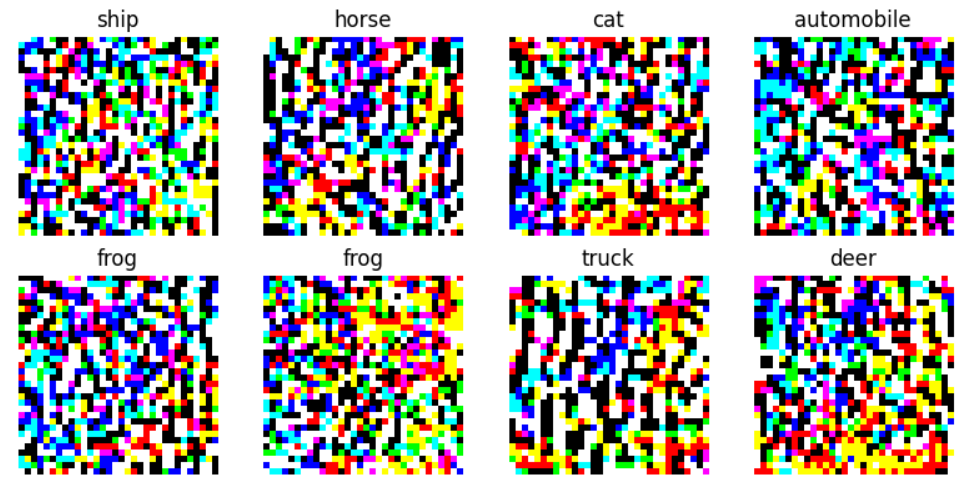}
  \caption{FLIP}
  \label{fig:is}
  \end{center}
\end{subfigure}%
\caption{\it \small Adversary's reconstructed training samples, in case of learning with (i) FLIP and (ii) DP-FL. FL settings is $n=5, K = 3, B = 8$.}
\label{fig:images}
\end{figure*}

\begin{figure}[h]
\centering
\includegraphics[width=0.7\linewidth]{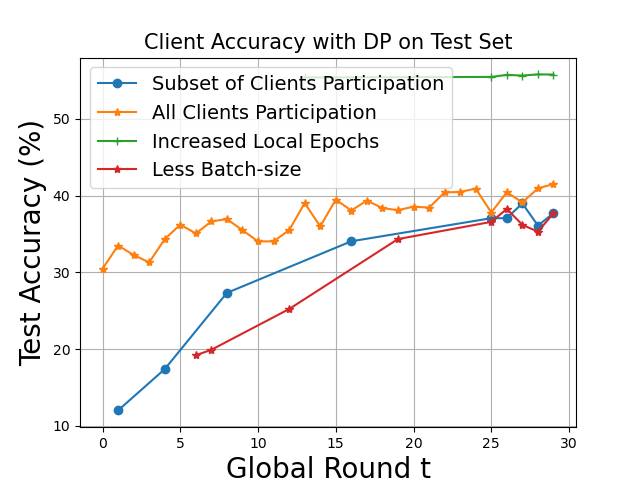}
\caption{\it \small CIFAR-10 test set accuracies of client's true model $\xc_t$, when LeNet-5 is trained using DP-FL. FL settings are the same as Fig.~\ref{fig:advclient}.}
\label{fig:dp_test}
\end{figure}

In summary, the first experiment shows that uploading the model updates $\xi_t$ (FLIP) to the server offers protection of the client's model against the eavesdropping adversary, as evidenced by the adversary's estimated model's evaluation on the CIFAR-10 test set. Here, the adversary's test accuracy of $10\%-12\%$ at steady-state indicates that the adversary's model performed as good as an untrained model, since random prediction on $10$ classes will also have $10\%$ accuracy on average. The second experiment considers accuracy of adversary's model evaluated on reconstructed client data as a metric. This experiment quantifies the adversary's estimated model on reconstructed client data in those scenarios where the adversary does not have direct access to actual samples from the data distribution. In the same settings, FLIP turns out to be much more effective in offering protection of the federated learning process than DP-FL without affecting the utility of federated learning.

\section{Conclusion}

In this work, we analyzed the protection of FL algorithms against eavesdropping adversaries.
Our results show that algorithms in which clients share model increments (e.g., SCAFFOLD) provide inherent protection, while those sending full local models do not.
We demonstrated how protection is influenced by client selection probability, adversary capability, and algorithm design.
Notably, algorithms using model increments retain protection even when the adversary intercepts client model updates at every round.
Simulations support our theoretical findings and highlight the advantages of this simple approach--i.e., sending model increments instead of local model--over differential privacy-based methods.

\bibliographystyle{ieeetr}
\bibliography{ref_imported, refs_imported, ref_fl}

\appendix

\subsection{Some Useful Results} \label{AP:someUsefulResults}

\begin{proposition} \label{prop:L}
   For any symmetric matrix $\Sigma$, let us define the linear operator $  \mathcal{L} (\Sigma) =\E[A_t \Sigma A_t^\intercal] $. Then, 
    \begin{align*}
       \mathcal{L} (\Sigma) = &  (1-p) \Sigma + p K_1\Sigma K_1^\intercal +p\gamma(1-\gamma) K_2 \Sigma K_2^\intercal. 
    \end{align*}
\end{proposition}

\begin{proof}
    Let us write $A_t =  \begin{bmatrix}
        I & \delta_t(1-\mu_t)M\\
        0 & (1-\delta_t) + \delta_t(1-\mu_t)M 
    \end{bmatrix}$ as 
    \begin{align*}
        A_t = \underbrace{\begin{bmatrix}
            I & 0\\
            0 & 0
        \end{bmatrix}}_{L_0} + \delta_t (1-\mu_t)  \underbrace{\begin{bmatrix}
            0 & M\\
            0 & M
        \end{bmatrix}}_{L_1} + (1-\delta_t) \underbrace{\begin{bmatrix}
            0 & 0\\
            0 & I
        \end{bmatrix}}_{L_2} .
    \end{align*}
    This results in 
    \begin{align*}
        \E[A_t \Sigma A_t^\intercal] = & L_0 \Sigma L_0^\intercal + p(1-\gamma)L_1 \Sigma L_1^\intercal + (1-p) L_2 \Sigma L_2^\intercal \\
        & + p(1-\gamma) (L_0 \Sigma L_1^\intercal + L_1 \Sigma L_0^\intercal) \\
        &+ (1-p)  (L_0 \Sigma L_2^\intercal + L_2 \Sigma L_0^\intercal),
    \end{align*}
    where we have used $\delta_t(1-\delta_t) = 0$, $\E[\delta_t^2] = \E[\delta_t] = p$, $\E[\mu_t^2] = \E[\mu_t] = \gamma$, and the independence of $\delta_t$ and $\mu_t$.
    The last equation can be re-grouped as 
    \begin{align*}
         \E[A_t \Sigma A_t^\intercal] = & p(L_0 + (1-\gamma)L_1)\Sigma (L_0 + (1-\gamma)L_1)^\intercal \\
         & + (1-p) (L_0 + L_2)\Sigma (L_0 + L_2)^\intercal \\
         & + p\gamma(1-\gamma)L_1 \Sigma L_1^\intercal.
    \end{align*}
    Notice that $L_0 + L_2 = I$. By defining 
    \begin{subequations}\label{eq:K1_K2}
    \begin{align} 
     K_1 =& L_0 + (1-\gamma)L_1 =    \begin{bmatrix}
            I & (1-\gamma)M\\
            0 & (1-\gamma)M
        \end{bmatrix},\\
    K_2 =& L_1 =      \begin{bmatrix}
            0 & M\\
            0 & M
        \end{bmatrix},
    \end{align}
    \end{subequations}
we may write 

    % \begin{align*}
      \noindent  \hspace{6pt }$\E[A_t \Sigma A_t^\intercal] = (1-p) \Sigma + p K_1\Sigma K_i^\intercal + p\gamma(1-\gamma) K_2 \Sigma K_2^\intercal$.
    % \end{align*}
\end{proof}

\begin{proposition} \label{prop:Vt}
    Suppose $v_t = \tilde A_t \bar \sigma_t + \tilde u_t$, as defined in \eqref{eq:sigma_tilde}. 
    Then, 
    \begin{align*}
        [I ~~0] \E[v_t v_t^\intercal] [I ~~0]^\intercal = p(1-p)(1-\gamma)(r_t+s_t)(r_t+s_t)^\intercal \\+ p^2\gamma(1-\gamma)s_t s_t^\intercal 
        +p(1-p)\gamma r_t r_t^\intercal \\ + p(1-\gamma)\E[(\E[\xi_{\tau_t}]- \xi_{\tau_t})(\E[\xi_{\tau_t}]- \xi_{\tau_t})^\intercal] ,
    \end{align*}
    where 
    \begin{align*}
        r_t &= \zeta_t - \zetahat_t,\\
        s_t &= \xi_t - M\E[\xi_{\tau_t}] + M \E[q_t]. 
    \end{align*}
\end{proposition}

\begin{proof}
    Verify that, 
    \begin{align*}
        [I~~0] v_t = \underbrace{(\delta_t(1-\mu_t) - p(1-\gamma))}_{:=\alpha_t} s_t + (\delta_t - p) r_t \\
        + \delta_t(1-\mu_t) (\E[\xi_{\tau_t}]- \xi_{\tau_t}).
    \end{align*}
    Therefore,
    \begin{align*}
        [I ~~0] \E[v_t& v_t^\intercal]  [I ~~0]^\intercal = \E[\alpha_t^2] s_t s_t^\intercal + \E[(\delta_t -p)^2]r_t r_t^\intercal \\
        &+ \E[\delta_t^2(1-\mu_t)^2]\E[(\E[\xi_{\tau_t}]- \xi_{\tau_t})(\E[\xi_{\tau_t}]- \xi_{\tau_t})^\intercal]\\
        & + \E[\alpha_t(\delta_t-p)](s_t r_t^\intercal + r_ts_t^\intercal), 
    \end{align*}
    where cross terms such as $\E[\alpha_t\delta_t(1-\mu_t)s_t(\E[\xi_{\tau_t}]- \xi_{\tau_t})]$ becomes zero since $\alpha_t,\delta_t$, and $\mu_t$ are independent of $\tau_t$, and hence, 
    $\E[\alpha_t\delta_t(1-\mu_t)s_t(\E[\xi_{\tau_t}]- \xi_{\tau_t})] = \E[\alpha_t\delta_t(1-\mu_t)]s_t\E[(\E[\xi_{\tau_t}]- \xi_{\tau_t})]= 0$.
    At this point, also verify that 
    \begin{align*}
        &\E[\alpha_t^2] = p(1-\gamma)(1-p + p\gamma),~~\\
        &\E[(\delta_t-p)^2] = p(1-p)\\
        &\E[\delta_t^2(1-\mu_t)^2] = p(1-\gamma),~~\\
        &\E[\alpha_t(\delta_t -p)]=p(1-p)(1-\gamma),
    \end{align*}
    and thus, after some rearrangement, 
    \begin{align*}
        V_t = p(1-p)(1-\gamma)(s_t+r_t)(s_t+r_t)^\intercal + p^2\gamma(1-\gamma)s_ts_t^\intercal \\
        +p(1-p)\gamma r_t r_t^\intercal + p(1-\gamma)\E[(\E[\xi_{\tau_t}]- \xi_{\tau_t})(\E[\xi_{\tau_t}]- \xi_{\tau_t})^\intercal]. 
    \end{align*}
    This completes the proof.
\end{proof}

\end{document}